\documentclass[11pt]{article}

\usepackage{amsmath}
\usepackage{amssymb}
\usepackage{amsthm}
\usepackage{mathtools}
\usepackage{hyperref}
\usepackage{fullpage}

\newtheorem{theorem}{Theorem}
\newtheorem{lemma}{Lemma}
\newtheorem{corollary}{Corollary}
\newtheorem{claim}{Claim}
\newtheorem{definition}{Definition}

\newenvironment{clmproof}[1][Proof]{\begin{proof}[#1]}{\end{proof}}

\renewcommand{\vec}[1]{\boldsymbol{#1}}

\newcommand{\F}{\mathbb{F}}

\newcommand{\N}{\mathbb{N}}

\title{\bf New Constructions of SD and MR Codes over Small Finite Fields}

\newcommand{\email}[1]{\href{mailto:#1}{\nolinkurl{#1}}}

\author{Guangda Hu\\
Princeton University\\
\email{guangdah@cs.princeton.edu}
\and
Sergey Yekhanin\\
Microsoft Research\\
\email{yekhanin@microsoft.com}
}

\date{}

\begin{document}

\maketitle

\begin{abstract}
Data storage applications require erasure-correcting codes with prescribed sets of dependencies between data symbols and redundant symbols. The most common arrangement is to have $k$ data symbols and $h$ redundant symbols (that each depends on all data symbols) be partitioned into a number of disjoint groups, where for each group one allocates an additional (local) redundant symbol storing the parity of all symbols in the group. A code as above is maximally recoverable, if it corrects all erasure patterns that are information theoretically correctable given the dependency constraints. A slightly weaker guarantee is provided by SD codes.

One key consideration in the design of MR and SD codes is the size of the finite field underlying the code as using small finite fields facilitates encoding and decoding operations. In this paper we present new explicit constructions of SD and MR codes over small finite fields.

\end{abstract}

\section{Introduction}

Consider a systematic linear $[n,k]$ code with codeword length $n=k+h+\frac{k+h}{r}$ for some integers $r$ and $h.$ Assume that the code has the following structure. There are $h$ redundant codeword coordinates (heavy symbols) that depend on all systematic coordinates (data symbols). Further, these $k+h$ coordinates are partitioned into $g=\frac{k+h}{r}$ sets of size $r,$ where for each set one allocates an additional (local) redundant symbol storing the parity of all symbols in the set. We refer to symbols in a set and their respective local parity as a local group.

Local codes with parameters $(k,r,h)$ as above have been recently studied~\cite{CHL,GHSY,BHH} and used in practice~\cite{HuangSX} in the context of erasure coding for data storage, where local parities facilitate fast recovery of any single symbol when it is erased, while heavy parities provide tolerance to a large number of simultaneous erasures.

A local code is Maximally Recoverable (MR) (equivalently, PMDS using the terminology from \cite{BHH}), if it corrects all erasure patterns which are information theoretically correctable given the prescribed  dependency relations between data symbols and parity symbols. This amounts to correcting every pattern of simultaneous erasures that can be obtained by erasing one symbol per local group and $h$ more arbitrary symbols. A somewhat weaker guarantee is provided by SD codes~\cite{Blaum,SD_fail}. Here one assumes that $r+1$ symbols in each of $g$ groups are ordered. A code is called SD if it corrects every pattern of simultaneous erasures that can be obtained by erasing the $i$-th symbol in every local group (for some arbitrary fixed $i\in [r+1]$) and $h$ more symbols.

In applications one is interested in explicit MR (or at least SD) codes defined over small finite fields, as the size of the field underlying the code determines computational efficiency of encoding and decoding and affects the throughput of the system. Constructing MR and SD codes over small finite fields has been a subject of a line of work.

Explicit families of MR local codes with $h=1$ and $h=2$ were obtained in~\cite{Blaum,BHH}. SD codes were introduced and studied in~\cite{Blaum,SD_fail}. Some non-explicit constructions of SD codes with $h=3,$ without analysis of the field size were given in~\cite{CSC}.

The first explicit families of MR local codes for all values of $k,r$ and $h$ were given in~\cite{GHJY14}. In the setting of $h=O(1),$ $r=O(1),$ and growing $k,$ these constructions yield field of size roughly $q=O\left(n^{h-1}\right).$ For $h=3,$ one gets a field of size $q=O\left(n^{1.5}\right).$ In contrast to this, when $h=2,$ $r=O(1),$ and $k$ grows, constructions of~\cite{BHH,GHJY14} yield codes over a field of optimal size $O(n).$ In the setting of $h=O(1),$ $g=O(1),$ and growing $k,$ constructions of~\cite{GHJY14} yield field of size $n^{(g+h)/2}.$

\subsection{Our results}
In this work we present two new explicit constructions of SD and MR codes over small finite fields. Our codes improve upon earlier results both in concrete settings and asymptotically. To keep the statements simple, we mainly focus on the asymptotic setting:
\begin{itemize}
\item We obtain a new family of $(k,r,h)$-SD codes with three heavy parities that uses a field of size $O(n)$ when $r=O(1)$ and $k$ grows. This shows that optimal linearly-growing field size is attainable not just for $h\leq 2$ but also for $h=3,$ at least in the SD model.

\item We give a new general construction of $(k,r,h)$-local MR codes. Our construction improves upon codes of~\cite{GHJY14} in the narrow setting of two local groups ($g=2$) when $h$ is a constant divisible by $4$ and $k$ grows. In this setting we get a field of size $n^{h/2}.$

Perhaps more importantly, unlike all previously known constructions that work for all $h,$ (with an exception of~\cite{TPD} that uses $n^{O(k)}$ field size) our code family is ``Vandermonde type'', rather than ``Linearized'', i.e., it uses consecutive exponents $1,2,3,\ldots,$ rather than $1,2,4,8,\ldots$ to define the parity check matrix. This is an important property as one can show that no ``Linearized'' construction can beat the $q=O\left(n^{h/2}\right)$ bound for the field size. New techniques are of vital interest.
\end{itemize}

\section{Preliminaries}
\begin{definition}[$(k,r,h)$-local codes, \cite{GHJY14}]
Let $C$ be a linear systematic $[n,k]$ code defined over some finite field. We say that $C$ is a {\em $(k,r,h)$-local code} if:
\begin{itemize}
\item $r\mid(k+h)$ and $n=k+h+(k+h)/r$;
\item There are $h$ {\em heavy parity} symbols, where each heavy parity is a linear combination of all $k$ data symbols.
\item The collection of $k+h$ symbols (data and heavy parities) is partitioned into $g=\frac{k+h}{r}$ sets of size $r,$ where for each set one allocates a (local) redundant symbol. We refer to symbols in a set and their respective local parity as a {\em local group}. Local parity ensures that the sum of all symbols in a local group is zero.
\end{itemize}
\end{definition}

In this paper, we require $g,r\geq 2.$ We use $\mathcal{C}$ to denote the set of all $(k,r,h)$-local codes of length $n$. Every code in $\mathcal{C}$ has a parity check matrix in the following form:
\begin{equation} \label{eq:parity}
H=\left[\begin{array}{ccc|c|ccc}
1 & \ldots & 1 & & & & \\
& & & \ldots\ldots\ldots & & & \\
& & & & 1 & \ldots & 1 \\
\hline
v_{1,1,1} & \ldots & v_{1,r+1,1} & \ldots\ldots\ldots & v_{g,1,1} & \ldots & v_{g,r+1,1} \\
v_{1,1,2} & \ldots & v_{1,r+1,2} & \ldots\ldots\ldots & v_{g,1,2} & \ldots & v_{g,r+1,2} \\
& \vdots & & \vdots & & \vdots & \\
v_{1,1,h} & \ldots & v_{1,r+1,h} & \ldots\ldots\ldots & v_{g,1,h} & \ldots & v_{g,r+1,h}
\end{array}\right].
\end{equation}
The top part contains $g$ rows that are linear constraints for the local parities in each group. The bottom part contains $h$ rows that are linear constraints corresponding to the heavy parities. There are $g$ groups of columns, that we call {\em wide columns}. Each wide column contains $r+1$ columns.

\begin{definition}[maximally recoverable codes, \cite{GHJY14}]
A code $C\in\mathcal{C}$ is {\em maximally recoverable} iff for any set $E\subseteq[n]$, where $E$ is obtained by picking one coordinate from each local group, if we puncture the code in coordinates specified by $E$ we obtain a maximal distance separable code that encodes a message of length $k$ to a codeword of length $k+h$.
\end{definition}

Our proofs rely on the following standard lemma:
\begin{lemma}[\cite{GHJY14}]\label{lem:pf}
A code $C\in\mathcal{C}$ defined by a parity check matrix $H$ (\ref{eq:parity}) is maximally recoverable, iff any set $T$ of $g+h$ columns of $H$ that is obtained by picking one column from each wide column and $h$ additional columns has full rank.
\end{lemma}

We are also interested in Sector-Disk (SD) codes \cite{Blaum}.
\begin{definition}[Sector-Disk code]
A code $C\in\mathcal{C}$ specified by a parity check matrix $H$ (\ref{eq:parity}) is an {\em SD code}, iff any set $T$ of $g+h$ columns of $H$ that is obtained by picking the $j$-th column from each wide column for some $j\in[r+1]$ and $h$ additional columns has full rank.
\end{definition}

\begin{definition}[$w$-independence]
Let $\F$ be a characteristic-2 finite field. We say a set $S\subseteq\F$ is $w$-independent if for all $T\subseteq S,$ $0<|T|\leq w,$ elements of $T$ do not sum to zero.
\end{definition}

Now we review a simple way of constructing MR codes, which was studied in~\cite{GHJY14} generalizing~\cite{BHH}. (Our Vandermonde-type construction in Section~\ref{Sec:VT} improves upon it in some regimes.)

\begin{theorem} \label{thm:easycons}
Suppose $r+1$ and $g$ are powers of $2$ (therefore $n$ is also a power of $2$). There is an explicit construction of $(k,r,h)$-local MR codes over a field of size $n^{(g+h)/2}$ when $g+h$ is even, or $2n^{(g+h-1)/2}$ when $g+h$ is odd.
\end{theorem}

\begin{proof}[(Proof sketch)]
In this setting the field size $n^{(g+h)/2}$ or $2n^{(g+h-1)/2}$ is a power of $2$. Let $v_{i,j,b}=x_{i,j}^{2^{b-1}}$ in~(\ref{eq:parity}), where $\{x_{i,j}\}$ are elements of the field, for all $i\in[g],j\in[r+1],b\in[h]$. By Proposition~10 in~\cite{GHJY14}, the code defined by the parity check matrix~(\ref{eq:parity}) is maximally recoverable if elements $\{x_{i,j}\}$ are ($g+h$)-independent. We construct these elements using BCH codes. If $g+h$ is even, we choose $\{x_{i,j}\}$ to have shape $\beta\circ\beta^3\cdots\circ\beta^{g+h-1}$; If $g+h$ is odd, we choose $\{x_{i,j}\}$  to have shape $1\circ\beta\circ\beta^3\cdots\circ\beta^{g+h-2}$, where $\beta$ runs through $\F_n$ and $\circ$ denotes concatenation of binary strings. One can verify the properties of the construction.
\end{proof}

\section{The construction of SD codes with $h=3$}

We now present our construction of SD codes with 3 heavy parities over a characteristic-2 field of size $O(r^3n).$ When $r$ is constant, the field size is linear of $n.$

By increasing each of $r$ and $n$ by at most a constant multiplicative factor we can have $(r+1)$ and $n$ be powers of two. Furthermore, by increasing $n$ by at most a multiplicative factor of $(r+1)$ we can have $\log_2 (r+1) \mid \log_2 n$. Consider the parity check matrix $H$ (\ref{eq:parity}), which has dimension $(g+3)\times n$. We set
\begin{equation} \label{eq:sdcons}
v_{i,j,1}=x_{i,j},\quad v_{i,j,2}=x_{i,j}^2,\quad v_{i,j,3}=x_{i,j}^4,
\end{equation}
where $i\in[g],j\in[r+1]$. To complete specifying $H$, we need to specify $\{x_{i,j}\},$ $i\in [g],$ $j\in [r+1]$ in some field $F_{2^t}.$ We set $t=2\log_2(r+1)+1+\log_2 n$. Let $S=\{s_1,\ldots,s_{r+1}\} \subseteq \mathbb{F}_{2(r+1)^2}$ be an ordered $5$-independent set. Such a set can easily be obtained from BCH codes. We can set the elements in $S$ to be $1\circ\beta_i\circ\beta_i^3$, where $\beta_i$ ($i\in[r+1]$) takes every element of $\F_{r+1}$, and $\circ$ denotes concatenation of binary strings. Consider the field $\mathbb{F}_n.$ Note that $\mathbb{F}_{r+1}=\{f_1,\ldots,f_{r+1}\}\subseteq \mathbb{F}_n.$ Let $\{\alpha_1,\ldots,\alpha_g\}\subseteq \mathbb{F}_n$ be such that for $i\ne j\in [g]:$
$$\alpha_i\cdot\mathbb{F}_{r+1}\cap\alpha_j\cdot\mathbb{F}_{r+1}=\{0\}.$$
The following formula specifies $\{x_{i,j}\}, i\in [g], j\in [r+1]$ in the field $F_{2^t}$ via their representation as bit strings:
\begin{equation}\label{Eqn:Xij}
x_{i,j}=s_j\circ \alpha_i f_j,
\end{equation}
where $\circ$ denotes concatenation of binary strings. Note that $|\mathbb{F}_{2^t}|=O(r^2n)$, or $O(r^3n)$ if we take the original transformation that we have applied to $r$ and $n$ into account.

\section{The proof of the SD construction}

The following theorem implies that matrix $H$ constructed in the previous section is a parity check matrix of an SD code.

\begin{theorem}\label{Th:SD_h3}
Let $j_1\in [r+1]$ be arbitrary. Consider a collection $T$ of $g+3$ columns of $H$ that is obtained by including all $g$ columns labeled by $x_{i,j_1},$ $i\in [g]$ as well as three additional columns. We claim that the matrix $T$ has full rank.
\end{theorem}
\begin{proof}
First note that we can safely discard all columns in $T$ that uniquely originate from their respective wide columns (local groups), as every such column has support on some coordinate where no other coordinate of $T$ does. This leaves us with three cases:

\textit{\textbf{Additional columns are from one wide column:}} We need to argue that any matrix of the form
$$
\left[
\begin{array}{cccc}
1           & 1           & 1           & 1           \\
x_{i,j_1}   & x_{i,j_2}   & x_{i,j_3}   & x_{i,j_4}   \\
x_{i,j_1}^2 & x_{i,j_2}^2 & x_{i,j_3}^2 & x_{i,j_4}^2 \\
x_{i,j_1}^4 & x_{i,j_2}^4 & x_{i,j_3}^4 & x_{i,j_4}^4 \\
\end{array}
\right]
$$
has full rank, where $i\in [g],$ and distinct $\{j_1,j_2,j_3,j_4\} \subseteq [r+1].$ By adding the first column to the other 3 columns, one can see that it suffices to show
$$
\left[
\begin{array}{ccc}
x_{i,j_1} + x_{i,j_2}     & x_{i,j_1} + x_{i,j_3}     & x_{i,j_1} + x_{i,j_4}     \\
(x_{i,j_1} + x_{i,j_2})^2 & (x_{i,j_1} + x_{i,j_3})^2 & (x_{i,j_1} + x_{i,j_4})^2 \\
(x_{i,j_1} + x_{i,j_2})^4 & (x_{i,j_1} + x_{i,j_3})^4 & (x_{i,j_1} + x_{i,j_4})^4 \\
\end{array}
\right]
$$
is non-degenerate. By the standard properties of finite fields \cite[Lemma~3.51]{LN83} this amounts to showing that no non-empty subset of elements of
\begin{equation}\label{Eqn:S}
\{ x_{i,j_1} + x_{i,j_2}, x_{i,j_1} + x_{i,j_3}, x_{i,j_1} + x_{i,j_4} \}
\end{equation}
sums to zero. To see that note that after trivial cancellations every sum of elements of~(\ref{Eqn:S}) involves between two and four distinct elements $x_{i,j_s}.$ Thus by formula~(\ref{Eqn:Xij}) and $5$-independence property of prefixes ${s_j}$ the sum is non-zero.

\textit{\textbf{Additional columns are from two wide columns:}}
Here we need to argue that any matrix of the form
$$
\left[
\begin{array}{ccccc}
1             & 1             & 1             &               &               \\
              &               &               & 1             & 1             \\
x_{i_1,j_1}   & x_{i_1,j_2}   & x_{i_1,j_3}   & x_{i_2,j_1}   & x_{i_2,j_4}   \\
x_{i_1,j_1}^2 & x_{i_1,j_2}^2 & x_{i_1,j_3}^2 & x_{i_2,j_1}^2 & x_{i_2,j_4}^2 \\
x_{i_1,j_1}^4 & x_{i_1,j_2}^4 & x_{i_1,j_3}^4 & x_{i_2,j_1}^4 & x_{i_2,j_4}^4 \\
\end{array}
\right]
$$
has full rank, where $i_1\ne i_2\in [g],$ $\{j_1,j_2,j_3\} \subseteq [r+1]$ are distinct, and $\{j_1,j_4\} \subseteq [r+1]$ are distinct. By adding the first column to the second and third columns, and adding the fourth column to the fifth column, one can see that it suffices to show
$$
\left[
\begin{array}{ccc}
x_{i_1,j_1} + x_{i_1,j_2}     & x_{i_1,j_1} + x_{i_1,j_3}     & x_{i_2,j_1} + x_{i_2,j_4}     \\
(x_{i_1,j_1} + x_{i_1,j_2})^2 & (x_{i_1,j_1} + x_{i_1,j_3})^2 & (x_{i_2,j_1} + x_{i_2,j_4})^2 \\
(x_{i_1,j_1} + x_{i_1,j_2})^4 & (x_{i_1,j_1} + x_{i_1,j_3})^4 & (x_{i_2,j_1} + x_{i_2,j_4})^4 \\
\end{array}
\right]
$$
is non-degenerate. This amounts to showing that no non-empty subset of elements of
\begin{equation}\label{Eqn:S2}
\{ x_{i_1,j_1} + x_{i_1,j_2}, x_{i_1,j_1} + x_{i_1,j_3}, x_{i_2,j_1} + x_{i_2,j_4} \}
\end{equation}
sums to zero. Restricting our attention to $(2\log_2 (r+1)+1)$-long prefixes~(\ref{Eqn:Xij}) of the elements above yields the collection
$$ \{ s_{j_1} + s_{j_2}, s_{j_1} + s_{j_3}, s_{j_1} + s_{j_4} \}. $$
It is easy to see that the only zero sums of the elements above are the first and the third elements (when $j_2=j_4$), or the second and the third elements (when $j_3=j_4$). Neither case however yields a zero sum of the respective elements of~(\ref{Eqn:S2}), since $(\log n)$-long suffixes of both $x_{i_1,j_1} + x_{i_1,j_2}$ and $x_{i_1,j_1} + x_{i_1,j_3}$ are non-zero elements in $\alpha_{i_1}\cdot \mathbb{F}_{r+1}$ while the $(\log n)$-long suffix of $x_{i_2,j_1} + x_{i_2,j_4}$ is a non-zero element in $\alpha_{i_2}\cdot \mathbb{F}_{r+1}.$

\textit{\textbf{Additional columns are from three wide columns:}}
Here we need to argue that any matrix of the form
$$
\left[
\begin{array}{cccccc}
1             & 1             &               &               &               &               \\
              &               &  1            & 1             &               &               \\
              &               &               &               & 1             & 1             \\
x_{i_1,j_1}   & x_{i_1,j_2}   & x_{i_2,j_1}   & x_{i_2,j_3}   & x_{i_3,j_1}   & x_{i_3,j_4}   \\
x_{i_1,j_1}^2 & x_{i_1,j_2}^2 & x_{i_2,j_1}^2 & x_{i_2,j_3}^2 & x_{i_3,j_1}^2 & x_{i_3,j_4}^2 \\
x_{i_1,j_1}^4 & x_{i_1,j_2}^4 & x_{i_2,j_1}^4 & x_{i_2,j_3}^4 & x_{i_3,j_1}^4 & x_{i_3,j_4}^4 \\
\end{array}
\right]
$$
has full rank, where $\{i_1,i_2,i_3\}\subseteq [g]$ are distinct; $j_1,j_2,j_3,j_4\in[r+1]$, $j_1\neq j_2$, $j_1\neq j_3$ and $j_1\neq j_4$. By adding the first column to the second, the third column to the fourth, and the fifth column to the sixth, one can see that it suffices to show
$$
\left[
\begin{array}{ccc}
x_{i_1,j_1} + x_{i_1,j_2}     & x_{i_2,j_1} + x_{i_2,j_3}     & x_{i_3,j_1} + x_{i_3,j_4}     \\
(x_{i_1,j_1} + x_{i_1,j_2})^2 & (x_{i_2,j_1} + x_{i_2,j_3})^2 & (x_{i_3,j_1} + x_{i_3,j_4})^2 \\
(x_{i_1,j_1} + x_{i_1,j_2})^4 & (x_{i_2,j_1} + x_{i_2,j_3})^4 & (x_{i_3,j_1} + x_{i_3,j_4})^4 \\
\end{array}
\right]
$$
is non-degenerate. This amounts to showing that no non-empty subset of elements of
\begin{equation}\label{Eqn:S3}
\{ x_{i_1,j_1} + x_{i_1,j_2}, x_{i_2,j_1} + x_{i_2,j_3}, x_{i_3,j_1} + x_{i_3,j_4} \}
\end{equation}
sums to zero. Restricting our attention to $(2\log_2 (r+1)+1)$-long prefixes~(\ref{Eqn:Xij}) of the elements above yields the collection
$$ \{ s_{j_1} + s_{j_2}, s_{j_1} + s_{j_3}, s_{j_1} + s_{j_4} \}. $$
It is easy to see that the zero sum has to involve exactly two elements. However, no sum involving two elements of (\ref{Eqn:S3}) can be zero since the $(\log n)$-long suffixes of $x_{i_1,j_1} + x_{i_1,j_2},$ $x_{i_2,j_1} + x_{i_2,j_3},$ and $x_{i_3,j_1} + x_{i_3,j_4}$ are non-zero elements in (respectively) $\alpha_{i_1}\cdot \mathbb{F}_{r+1},$ $\alpha_{i_2}\cdot \mathbb{F}_{r+1},$ and $\alpha_{i_3}\cdot \mathbb{F}_{r+1}.$ \qedhere
\end{proof}

\section{Vandermonde-type construction of MR codes}\label{Sec:VT}
We now present our new general construction of $(k,r,h)$-local MR codes. We note that this new construction does not follow the paradigm of only using exponents $1, 2, 4, \ldots.$
By increasing $n$ by at most a constant multiplicative factor we can have $n$ be a prime power. Let $\F_q$ be the field that we are working on. We pick $q$ to be a power of $n.$ Thus $\F_q$ is an extension of $\F_n.$ Let $t$ be a parameter to be determined later. Our construction uses a field of size $q=n^{h+g-t}.$

Let $\alpha\in\F_q$ be such that every element of $\F_q$ can be uniquely represented as $\lambda_0+\lambda_1\alpha+\cdots\lambda_{h+g-t-1}\alpha^{h+g-t-1},$ where $\lambda_0,\ldots,\lambda_{h+g-t-1}\in\F_n.$ We partition $\F_n$ into disjoint sets $S_1,\ldots,S_g$, each of size $r+1$. Let $x_{i,1},\ldots,x_{i,r+1}$ be elements of $S_i$ ($i\in[g]$). We set $v_{i,j,b}$ ($i\in[g], j\in[r+1],b\in[h]$) in the matrix $H$~(\ref{eq:parity}) as follows:

\begin{equation} \label{eq:v}
v_{i,j,b}=
\begin{cases}
x_{i,j}^b & b\leq t-1, \\
\langle\vec{u}_{b-t+1},\vec{w}_{x_{i,j}}\rangle & b\geq t,
\end{cases}
\end{equation}
where $\vec{w}_x$ denotes the vector $(x^t,x^{t+1},\ldots,x^{h+g-1})^T$ and $\vec{u}_1,\ldots,\vec{u}_{h-t+1}\in\F_q^{h+g-t}$ are linearly independent vectors satisfying
\begin{equation} \label{eq:u}
A\cdot\vec{u}_\ell=0,
\end{equation}
where $A$ denotes the matrix $\{A_{ij}=\alpha^{(j-1)n^{i-1}}\}_{(g-1)\times(h+g-t)}$, and $\ell\in[h-t+1]$. Note that there are $g-1$ rows in $A$ and we can always find $(h+g-t)-(g-1)=h-t+1$ linearly independent vectors $\vec{u}_\ell$ satisfying the requirement~(\ref{eq:u}).

We now prove the following theorems. In Theorem~\ref{thm:general}, we show that our construction gives MR codes with field size $q=n^{\lfloor(1-\frac{1}{g})h\rfloor+g-1}$. Then in Theorem~\ref{thm:improve}, we improve this result by choosing a different value of $t$ when the parameters satisfy certain conditions.
\begin{theorem} \label{thm:general}
Setting $t=\lceil\frac{h}{g}\rceil+1$, under the condition that $n$ is a prime power, the matrix $H$ defined in the above construction is a parity check matrix of an MR code with field size $q=n^{h+g-t}=n^{\lfloor(1-\frac{1}{g})h\rfloor+g-1}$.
\end{theorem}

\begin{theorem} \label{thm:improve}
Let $p$ be a prime. Suppose $g,r+1$ are powers of $p$ (therefore $n$ is also a power of $p$), $h\not\equiv1\pmod g$ and $\lceil\frac{h}{g}\rceil\not\equiv p-1\pmod p$.

Let $S$ be an additive subgroup of $\F_n,$ $|S|=r+1$, and let $S_1,\ldots,S_g$ in the above construction to be shifts of $S$ ($S_i=S+\delta_i$ for some $\delta_i\in\F_n$) so that $\F_n=S_1\sqcup\cdots\sqcup S_g.$

Let $t=\lceil\frac{h}{g}\rceil+2$. The matrix $H$ above is a parity check matrix of an MR code over a field of size $q=n^{h+g-t}=n^{\lfloor(1-\frac{1}{g})h\rfloor+g-2}$.
\end{theorem}

Setting $p=2,$ from Theorem~\ref{thm:improve} we have
\begin{corollary}
For $g=2$ and $h\equiv0\pmod4$, there is an explicit construction of MR codes with field size $q=n^{h/2}$ for $n$ that is a power of $2$.
\end{corollary}

We note that this corollary beats Theorem~\ref{thm:easycons} for the case $g=2$ and $h\equiv0\pmod4$.

\section{Proofs for Vandermonde-type construction}

In this Section, we prove Theorems~\ref{thm:general} and~\ref{thm:improve} using Lemma~\ref{lem:pf}. We consider a collection $T$ of $g+h$ columns of $H$ that is obtained by including one column from each wide column and $h$ additional columns. It suffices to prove that $T$ has full rank.

We perform two operations on $T$: 1) Discard every wide column that contains only one column, and the first row at which that column has a $1.$ Let $g'\in[g]\cap[h]$ be the number of remaining wide columns. 2) Add the first $g'-1$ rows to the $g'$-th row. We obtain a matrix of size $(g'+h)\times(g'+h).$

Note that in the construction, the order of local groups and the order of columns in each local groups are arbitrary. Without loss of generality, we assume the matrix obtained via 1) and 2) is:
$$
M=\left[\begin{array}{c|c|c|c}
1~\cdots~1~ & & & \\
& 1~\dots~1~ & & \\
& & \ddots & \\
1~\cdots~1~ & 1~\cdots~1~ & 1~\cdots~1 & 1~\cdots~1~~ \\
\vec{v}_{1,1} \cdots \vec{v}_{1,r_1} & \vec{v}_{2,1} \cdots \vec{v}_{2,r_2} & \cdots~\cdots & \vec{v}_{g',1} \cdots \vec{v}_{g',r_{g'}}
\end{array}\right],
$$
where $\vec{v}_{i,j}$ denotes $(v_{i,j,1},\ldots,v_{i,j,h})^T$, $r_1,\ldots,r_{g'}\in[2,r+1]$ and $r_1+\cdots+r_{g'}=h+g'$. We need to prove $\det(M)\neq0$.

\begin{theorem} \label{thm:step}
For any $t\in[2,h]$, if $\det(M)=0;$ then there has to exist a polynomial $f(x)$ with $1\leq\deg(f)\leq t-1$ and $\mu_1,\ldots,\mu_{g'}\in\F_q$ such that $f(x_{i,j})=\mu_i$ for all $i\in[g'],j\in[r_i].$
\end{theorem}

\begin{proof}
Let $\vec{z}=(\mu_1,\mu_2,\ldots,\mu_{g'-1},\lambda_0,\lambda_1,\ldots,\lambda_h)\in\F_q^{h+g'}$ be a non-zero vector such that $\vec{z}M=\vec{0}^T$, and let $\mu_{g'}=0$. From the construction of $v_{i,j,b}$ (\ref{eq:v}), we can see that the following polynomial $f(x)$ has degree at most $h+g-1$ and satisfies $f(x_{i,j})=\mu_i$ for $i\in[g'],j\in[r_i]$:
$$
f(x)=\sum_{b=0}^{t-1}\lambda_bx^b+\sum_{b=t}^h\lambda_b\langle\vec{u}_{b-t+1},\vec{w}_x\rangle=\sum_{b=0}^{t-1}\lambda_bx^b+\sum_{b=t}^{h+g-1}c_bx^b
$$
where $\vec{w}_x$ denotes the vector $(x^t,\ldots,x^{h+g-1})^T$, and we use $(c_t,\ldots,c_{h+g-1})^T=\lambda_t\vec{u}_1+\cdots+\lambda_h\vec{u}_{h-t+1}$ to denote the coefficients of $x^t,\ldots,x^{h+g-1}$ in $f(x)$. By the constructions of $\vec{u}_1,\ldots,\vec{u}_{h-t+1}$~(\ref{eq:u}), we can see
\begin{equation} \label{eq:coef}
A\cdot(c_t,\ldots,c_{h+g-1})^T=\vec{0}.
\end{equation}

Next, we show that $1\leq\deg(f)\leq t-1.$

\begin{claim}
$f(x)$ is not a constant, i.e., $\deg(f)\geq1$.
\end{claim}

\begin{clmproof}[Proof of Claim 1]
In the construction we have ensured that $\vec{u}_1,\ldots,\vec{u}_{h-t+1}$ are linearly independent. So the coefficients $c_t,\ldots,c_{h+g-1}$ are all zeros if and only if $\lambda_t=\cdots=\lambda_h=0$. If $f(x)$ is a constant, we have $\lambda_1=\cdots=\lambda_{t-1}=0$ and $c_t=\cdots=c_{h+g-1}=0$. Hence $\lambda_1=\cdots=\lambda_h=0$, and by $\vec{z}M=\vec{0}^T$, the first $g'$ rows of $M$ are linearly dependent (with coefficients $\mu_1,\ldots,\mu_{g'-1},\lambda_0$), which is clearly false.
\end{clmproof}

For a list of $\F_q$ elements $(a_1,\ldots,a_m)\in\F_q^m$, we define the {\em $\F_n$ dimension} of the list as the dimension of these $\F_q$ elements when they can be linearly combined with coefficients in $\F_n$. We use $\dim_{\F_n}(a_1,\ldots,a_m)$ to denote this dimension.

\begin{claim}
$\dim_{\F_n}\{\lambda_0,\ldots,\lambda_{t-1},c_t,\ldots,c_{h+g-1}\}\leq g-1$.
\end{claim}

\begin{clmproof}[Proof of Claim 2]
Using Lagrange interpolating polynomials, we can find a polynomial $\psi(x)$ that agrees with $f(x)$ on $h+g'$ different values $x_{i,j}$ ($i\in[g'],j\in[r_i]$):
$$\psi(x)=\sum_{i=1}^{g'}\sum_{j=1}^{r_i}\mu_i\frac{\prod_{(i',j')\neq(i,j)}(x-x_{i',j'})}{\prod_{(i',j')\neq(i,j)}(x_{i,j}-x_{i',j'})}.$$
For the case $g'=1$, the above $\psi(x)\equiv0$ since $\mu_{g'}=0$.

Note that $\deg(f)\leq h+g-1$ and $\deg(\psi)\leq h+g'-1$. If $g'=g$, we have $f(x)\equiv \psi(x).$ Since $\{x_{i,j}\}_{i\in[g],j\in[r+1]}$ are from $\F_n,$ every coefficient of $f(x)$ can be written as an $\F_n$ linear combination of $\mu_1,\ldots,\mu_{g-1}.$ Hence the $\F_n$ dimension of $f(x)$ coefficients is at most $g-1.$

Next we consider the case $g'<g$. Since $f(x)$ agrees with $\psi(x)$ on $x=x_{i,j}$ for all $i\in[g'],j\in[r_i]$, we can see that there exist $\nu_0,\ldots,\nu_{g-g'-1}\in\F_q$ such that
$$
f(x)\equiv\psi(x)+\left(\sum_{i=0}^{g-g'-1}\nu_ix^i\right)\cdot\prod_{i\in[g'],j\in[r_i]}(x-x_{i,j}).
$$
Therefore the coefficients of $f(x)$ are $\F_n$ linear combinations of $\mu_1,\ldots,\mu_{g'-1},\nu_0,\ldots,\nu_{g-g'-1}$. The $\F_n$ dimension of these coefficients is at most $g'-1+g-g'=g-1$.
\end{clmproof}

Let $d\leq g-1$ be the $\F_n$ dimension of $c_t,\ldots,c_{h+g-1}$.

\begin{claim}
$d=0$, i.e., $c_t=\cdots=c_{h+g-1}=0$, $\deg(f)\leq t-1$.
\end{claim}

\begin{clmproof}[Proof of Claim 3]
Assume $d>0$. Let $\{\beta_1,\ldots,\beta_d\}\in\F_q^d$ be a basis of $\{c_t,\ldots,c_{h+g-1}\}$ (in the sense that $F_q$ elements can be linearly combined with coefficients in $\F_n$). Then we have an $(h+g-t)\times d$ matrix $\Xi=\{\xi_{i,j}\}$ over $\F_n$ such that
$$(c_t,\ldots,c_{h+g-1})^T=\Xi\cdot(\beta_1,\ldots,\beta_d)^T,$$
and $\textrm{rank}(\Xi)=d$. By~(\ref{eq:coef}), we have
\begin{equation} \label{eq:coef2}
A\cdot\Xi\cdot(\beta_1,\ldots,\beta_d)^T=0.
\end{equation}
Let $\tau_j=\xi_{1,j}+\xi_{2,j}\alpha+\cdots+\xi_{h+g-t,j}\alpha^{h+g-t-1}$ ($j\in[d]$). Since $\xi_{i,j}\in\F_n$, for all $\ell\in\N$ we have $\xi_{i,j}^{n^\ell}=\xi_{i,j}$ and
$$\tau_j^{n^\ell}=\xi_{1,j}+\xi_{2,j}\alpha^{n^\ell}+\cdots+\xi_{h+g-t,j}\alpha^{(h+g-t-1)n^\ell}.$$
We can see
$$
A\cdot\Xi=\begin{bmatrix}
\tau_1 & \cdots & \tau_d \\
\tau_1^n & \cdots & \tau_d^n \\
& \vdots & \\
\tau_1^{n^{g-2}} & \cdots & \tau_d^{n^{g-2}}
\end{bmatrix}.
$$
The $d\times d$ submatrix (note that $d\leq g-1$) at the top part of this matrix has full rank if and only if $\tau_1,\ldots,\tau_d$ are linearly independent (with coefficients in $\F_n$) \cite[Lemma~3.51]{LN83}. By the choice of $\xi_{i,j}$ we can see that $\tau_1,\ldots,\tau_d$ are linearly independent. Hence the matrix $A\cdot\Xi$ has rank $d,$ contradicting~(\ref{eq:coef2}).
\end{clmproof}

Combining Claims~1 and~3, we have $1\leq\deg(f)\leq t-1$. This concludes the proof of Theorem~\ref{thm:step}.
\end{proof}

Using Theorem~\ref{thm:step}, we are able to prove Theorems~\ref{thm:general} and~\ref{thm:improve}.

\begin{proof}[Proof of Theorem~\ref{thm:general}]
Assume $\det(M)=0$. The average value of $r_1,\ldots,r_{g'}$ is
$$\frac{h+g'}{g'}=\frac{h}{g'}+1\geq\frac{h}{g}+1.$$
Assume $r_i\geq\lceil\frac{h}{g}\rceil+1=t$, where $i\in[g']$. By Theorem~\ref{thm:step}, there exists a polynomial $f(x)$ with $1\leq\deg(f)\leq t-1$ satisfying $f(x)=\mu_i$ for some $\mu_i\in\F_q$ and $r_i>t-1$ different values $x=x_{i,j}$ ($j\in[r_i]$). We arrive at a contradiction.
\end{proof}

\begin{proof}[Proof of Theorem~\ref{thm:improve}]
Assume $\det(M)=0$. If there exists $i\in[g']$ such that $r_i\geq\lceil\frac{h}{g}\rceil+2=t$, we derive a contradiction with Theorem~\ref{thm:step} along the lines of the  proof of Theorem~\ref{thm:general}. Thus we only consider the case $r_i\leq\lceil\frac{h}{g}\rceil+1,$ for all $i\in[g'].$ The average value of $r_1,\ldots,r_{g'}$ is at least $h/g+1$ as in the proof of Theorem~\ref{thm:general}. Therefore there exists some $i_0\in[g']$ with $r_{i_0}=\lceil\frac{h}{g}\rceil+1$. We claim that there has to exist a different $i_1\in[g']$ ($i_1\neq i_0$) with $r_{i_1}=\lceil\frac{h}{g}\rceil+1.$ If there were no such $i_1$ we would have:
$$h+g'=\sum_{i\in[g']}r_i\leq\Big(\lceil\frac{h}{g}\rceil+1\Big)+(g'-1)\lceil\frac{h}{g}\rceil=g'\lceil\frac{h}{g}\rceil+1$$
$$\Rightarrow h+g\leq g\cdot\lceil\frac{h}{g}\rceil+1\Rightarrow g-1\leq g\cdot\left(\lceil\frac{h}{g}\rceil-\frac{h}{g}\right).$$

The latter inequality holds only when $\lceil\frac{h}{g}\rceil-\frac{h}{g}$ achieves its maximum value $\frac{g-1}{g}$. In other words, this happens only when $h\equiv1\pmod g$. Hence under the assumption $h\not\equiv1\pmod g$, there exists $i_0\neq i_1\in[g']$ with $r_{i_0}=r_{i_1}=\lceil\frac{h}{g}\rceil+1=t-1$.

By Theorem~\ref{thm:step}, there exists a polynomial $f(x)$ with $1\leq\deg(f)\leq t-1$ such that $f(x)=\mu_{i_0}$ for some $\mu_{i_0}\in\F_q$ and $r_{i_0}=t-1$ different values $x=x_{i_0,j}$ ($j\in[r_{i_0}]$), and $f(x)=\mu_{i_1}$ for some $\mu_{i_1}\in\F_q$ and $r_{i_1}=t-1$ different values $x=x_{i_1,j}$ ($j\in[r_{i_1}]$). We can see that $f(x)$ can be written in two ways
\begin{align*}
f(x)&\equiv B_0(x-x_{i_0,1})(x-x_{i_0,2})\cdots(x-x_{i_0,t-1})+\mu_{i_0} \\
&\equiv B_1(x-x_{i_1,1})(x-x_{i_1,2})\cdots(x-x_{i_1,t-1})+\mu_{i_1},
\end{align*}
where $B_0,B_1\in\F_q$ are not zero. We consider the $x^{t-1}$ term in the expansions of these two representations, and conclude that $B_0=B_1.$ Then we consider the $x^{t-2}$ term. We have
\begin{equation} \label{eq:sumg}
x_{i_0,1}+x_{i_0,2}+\cdots+x_{i_0,t-1}=x_{i_1,1}+x_{i_1,2}+\cdots+x_{i_1,t-1}.
\end{equation}
However, the identity above cannot hold. To see this, note that $\lceil\frac{h}{g}\rceil\not\equiv p-1\pmod p,$ or $t-1=\lceil\frac{h}{g}\rceil+1\not\equiv0\pmod p$. For every $j\in[t-1]$, $x_{i_0,j}-x_{i_1,j}$ can be written as $y_j+\delta_{i_0}-\delta_{i_1}$, where $y_j\in S$, and $\delta_{i_0}-\delta_{i_1}\neq 0$. Since $t-1$ is not a multiple of $p$, we can see that $(t-1)\cdot(\delta_{i_0}-\delta_{i_1})$ (summing $\delta_{i_0}-\delta_{i_1}$ for $t-1$ times) is non-zero and
$$\sum_{j=1}^{t-1}(x_{i_0,j}-x_{i_1,j})=\left(\sum_{j=1}^{t-1}y_j\right)+(t-1)\cdot(\delta_{i_0}-\delta_{i_1})\notin S.$$
Since $0\in S$,~(\ref{eq:sumg}) does not hold. This concludes the proof of Theorem~\ref{thm:improve}.
\end{proof}

\bibliographystyle{alpha}
\bibliography{MRcons}

\newcommand{\etalchar}[1]{$^{#1}$}
\begin{thebibliography}{GHSY12}

\bibitem[BHH13]{BHH}
Mario Blaum, James~Lee Hafner, and Steven Hetzler.
\newblock Partial-{MDS} codes and their application to {RAID} type of
  architectures.
\newblock {\em IEEE Transactions on Information Theory}, 59(7):4510--4519,
  2013.

\bibitem[Bla13]{Blaum}
Mario Blaum.
\newblock Construction of {PMDS} and {SD} codes extending {RAID} 5.
\newblock Arxiv 1305.0032, 2013.

\bibitem[CHL07]{CHL}
Minghua Chen, Cheng Huang, and Jin Li.
\newblock On maximally recoverable property for multi-protection group codes.
\newblock In {\em IEEE International Symposium on Information Theory (ISIT)},
  pages 486--490, 2007.

\bibitem[CSYS15]{CSC}
Junyu Chen, Kenneth~W. Shum, Quan Yu, and Wan Sung.
\newblock Sector disk codes and partial {MDS} codes with up to three global
  parities.
\newblock In {\em IEEE International Symposium on Information Theory (ISIT)},
  2015.

\bibitem[GHJY14]{GHJY14}
Parikshit Gopalan, Cheng Huang, Bob Jenkins, and Sergey Yekhanin.
\newblock Explicit maximally recoverable codes with locality.
\newblock {\em {IEEE} Transactions on Information Theory}, 60(9):5245--5256,
  2014.

\bibitem[GHSY12]{GHSY}
Parikshit Gopalan, Cheng Huang, Huseyin Simitci, and Sergey Yekhanin.
\newblock On the locality of codeword symbols.
\newblock {\em IEEE Transactions on Information Theory}, 58(11):6925 --6934,
  2012.

\bibitem[HSX{\etalchar{+}}12]{HuangSX}
Cheng Huang, Huseyin Simitci, Yikang Xu, Aaron Ogus, Brad Calder, Parikshit
  Gopalan, Jin Li, and Sergey Yekhanin.
\newblock Erasure coding in {W}indows {A}zure {S}torage.
\newblock In {\em USENIX Annual Technical Conference (ATC)}, pages 15--26,
  2012.

\bibitem[LN83]{LN83}
Rudolf Lidl and Harald Niederreiter.
\newblock {\em Finite Fields}.
\newblock Cambridge University Press, 1983.

\bibitem[PBH13]{SD_fail}
James~S. Plank, Mario Blaum, and James~L. Hafner.
\newblock {SD} codes: Erasure codes designed for how storage systems really
  fail.
\newblock In {\em 11th USENIX Conference on File and Storage Technologies (FAST
  13)}, pages 95--104, San Jose, CA, February 2013.

\bibitem[TPD13]{TPD}
Itzhak Tamo, Dimitris~S. Papailiopoulos, and Alexandros~G. Dimakis.
\newblock Optimal locally repairable codes and connections to matroid theory.
\newblock In {\em IEEE International Symposium on Information Theory (ISIT)},
  pages 1814--1818, 2013.

\end{thebibliography}

\end{document}